\documentclass[12pt,journal,onecolumn,draftcls]{IEEEtran}
\usepackage[T1]{fontenc}
\usepackage[latin9]{inputenc}
\usepackage{verbatim}
\usepackage{float}
\usepackage{amsthm}
\usepackage{amsmath}
\usepackage{amssymb}
\usepackage{graphicx}
\usepackage{color}
\usepackage{url}
\usepackage{authblk}

\makeatletter

\numberwithin{equation}{section}
\numberwithin{figure}{section}

\theoremstyle{plain}
\newtheorem{thm}{\protect\theoremname}[section]
\theoremstyle{definition}
\newtheorem{defn}[thm]{\protect\definitionname}
\theoremstyle{remark}

\theoremstyle{plain}
\newtheorem{lem}[thm]{\protect\lemmaname}
\newtheorem*{lem*}{Lemma}
\theoremstyle{remark}

\theoremstyle{plain}
\newtheorem{corollary}[thm]{\protect\corollaryname}
\theoremstyle{plain}
\newtheorem{proposition}[thm]{\protect\propositionname}
\newtheorem{example}[thm]{\protect\examplenname}

\newtheorem{thmApp}{\protect\theoremname}[subsection]
\theoremstyle{plain}
\newtheorem{propositionApp}[thmApp]{\protect\propositionnameApp}
\providecommand{\propositionnameApp}{Proposition}

\providecommand{\claimname}{Claim}
\providecommand{\definitionname}{Definition}
\providecommand{\lemmaname}{Lemma}
\providecommand{\remarkname}{Remark}
\providecommand{\theoremname}{Theorem}
\providecommand{\corollaryname}{Corollary}
\providecommand{\propositionname}{Proposition}
\providecommand{\examplenname}{Example}

\newcommand{\CN}{\mathbb{C}^N}
\newcommand{\CW}{\mathbb{C}^W}

\newcommand{\SV}{\left\vert S\right\vert}
\newcommand{\I}{\iota}
\newcommand\R{\mathbb{R}}
\newcommand\C{\mathbb{C}}
\newcommand\neta{\eta}
\newcommand\rev[1]{{\color{black}#1}}

\begin{document}

\title{Blind Phaseless Short-Time Fourier Transform Recovery}

\author[1]{Tamir Bendory}
\author[2]{Dan Edidin}
\author[3]{Yonina C. Eldar}

\affil[1]{The Program in Applied and Computational Mathematics, Princeton University, Princeton, NJ, USA}
\affil[2]{Department of Mathematics, University of Missouri, Columbia, Missouri, USA}
\affil[3]{The Andrew and Erna Viterbi Faculty of Electrical Engineering, Technion - Israel Institute of Technology, Haifa, Israel}

%\date{}
\maketitle

\begin{abstract}

\normalsize The problem of recovering a pair of signals from their blind phaseless short-time Fourier transform measurements arises in several important  phase retrieval applications, including ptychography and ultra-short pulse characterization. 
In this paper, we prove that in order to determine \rev{a pair of generic} signals uniquely, up to  trivial ambiguities, the number of phaseless measurements one needs to collect is, at most, five times the number of parameters \rev{required to} describe the signals.
 This result improves significantly upon previous papers, which required the number of measurements to be quadratic in the number of parameters rather than linear. 

In addition, we consider the simpler problem of recovering a pair of \rev{generic} signals from their blind short-time Fourier transform, 
when the phases are known. In this setting, which can be understood as a special case of the blind deconvolution problem, we show that the number of measurements required to determine the two signals, up to trivial ambiguities, equals exactly  the number of parameters to be recovered. 

As a side result, we study the classical phase retrieval problem---that is, recovering a signal from its Fourier magnitudes---when some entries of the signal are known a priori. 
We derive a bound on the number of required measurements as a function of the size of the set of known entries. 
Specifically, we show that if  most of the signal's entries are known, then only a few Fourier magnitudes are necessary to determine a signal uniquely.
\end{abstract}


\section{Introduction}

Phase retrieval is the problem of recovering a signal from its Fourier magnitudes. This problem arises in a variety of applications and scientific fields, such as X-ray crystallography~\cite{harrison1993phase,giacovazzo2002fundamentals}, optical imaging~\cite{walther1963question,shechtman2015phase}, ultra-short pulse characterization~\cite{trebino2012frequency}, astronomy~\cite{fienup1987phase} and signal processing~\cite{baykal2004blind,rabiner1993fundamentals,bendory2017bispectrum}.  For recent surveys from a signal processing perspective; see~\cite{bendory2017fourier,shechtman2015phase,jaganathan2015phase}.

Evidently,  there are infinitely many signals with the same Fourier magnitudes.  Therefore, to make the problem well-posed---that is, having a unique mapping between the Fourier magnitudes and the signal---some additional information on the signal must be harnessed.  In many applications it is common to assume prior knowledge on the structure of the underlying signal. For instance, in crystallography the signal is  sparse~\cite{harrison1993phase,giacovazzo2002fundamentals,elser2017benchmark}. If the signal is known to be {of} minimum phase, then there is a unique mapping between the Fourier magnitudes and the signal~\cite{huang2016phase}.
Other useful assumptions are that the signal has nonzero entries only in a known region or that its entries are nonnegative.

We focus in this paper on an alternative strategy to enforce uniqueness by collecting additional phaseless measurements. One important example, which serves as the the main motivation for this paper, is  \emph{ptychography}~\cite{nellist1995resolution,rodenburg2008ptychography,maiden2011superresolution,bendory2018non}. In ptychography, the specimen (i.e., signal, image or volume) is scanned by a localized illumination beam and Fourier magnitudes of overlapping \rev{windowed measurements} are recorded. 
Another  popular technique that collects multiple measurements is Frequency-Resolved
 Optical Gating (FROG), which is used to  characterize ultra-short laser pulses~\cite{trebino2012frequency}. In FROG, the
Fourier magnitudes of the product of the signal with a shifted version of itself are recorded, for  several  shifts. 
An extension, called \emph{blind FROG}, can be used to characterize two signals simultaneously by measuring the Fourier magnitudes of the product of one signal with a shifted version of the other.

In ptychography and blind FROG, the collected data can be modeled as the  
phaseless blind short-time Fourier transform (STFT) with respect to two signals. A detailed mathematical model is provided in Section~\ref{sec:model}.
In ptychography, the two signals are the specimen and the illumination beam. In blind FROG, the signals are the two optical pulses to be characterized and the blind phaseless STFT measurements are usually called {the} blind FROG trace. In those applications, a variety of algorithms have been suggested  to estimate both signals simultaneously~\cite{thibault2009probe,maiden2009improved,guizar2008phase,kane2008principal,trebino2012frequency}. However, these methods lack   theoretical analysis.
In this paper, we study the question of uniqueness, characterize the trivial ambiguities of the problem and derive bounds on the number of measurements required to determine the 
 two signals uniquely, up to these ambiguities.
 
We begin by presenting the trivial ambiguities of the phaseless blind STFT problem  in Proposition~\ref{prop:symmetries_phaseless}. 
The main result of this paper is Theorem~\ref{thm:phaselessblindstft}. It shows that the number of measurements required to uniquely determine a generic pair of signals, up to trivial ambiguities, is (at most) five times the number of parameters that describe the signals.
This result significantly improves upon previous results~\cite{bendory2017uniqueness} that required the number of measurements to be quadratic in the number of parameters rather than linear. 

We also consider a simpler problem  of recovering two signals from their blind STFT, when the phases are assumed to be known. This problem can be understood as a special case of the \emph{blind deconvolution} problem~\cite{jefferies1993restoration,shalvi1990new,abed1997blind,ayers1988iterative,michaeli2014blind}; see {the} discussion in Section~\ref{sec:model}. 
We show that the number of measurements required to determine the two signals uniquely, up to scaling ambiguities, is optimal.  That is, the number of measurements equal exactly the number of parameters to be recovered; see Theorem~\ref{thm:uniqueness_withphases}. The dimension of the ambiguity group is inversely proportional to the overlap between adjacent sections. Hence, a small overlap  results in a large  scaling ambiguities group.
The technical details are presented and discussed in Section~\ref{sec:main_results}, while proofs are provided in Section~\ref{sec:proofs} and Appendix~\ref{sec:proof_step3uniqueness}. Some of the proofs require basic definitions in group theory, which are  summarized in Appendix~\ref{sec:groups}. 

As a side result of this work, we study in Section~\ref{sec:partial_known_signal}  the classical phase retrieval problem---that is, recovering a signal from its Fourier magnitudes---when some entries of the signal are known.  
In~\cite{beinert2017enforcing}, Beinert and Plonka showed that one entry of the signal together with its Fourier magnitudes determine almost all signals uniquely. In Proposition~\ref{prop:limited}, we extend this result by bounding the number of required phaseless Fourier measurements   when several entries are known.
In particular, we show that if the number of unknown entries is relatively small, then one needs to collect only few Fourier {intensity} measurements. 
We successively use this result  to determine the pair of  signals in a specified section of the phaseless blind STFT measurements   based on the Fourier magnitudes  of this section and some of the signals' entries.  
These known entries are shared with an adjacent overlapping section  whose entries were determined previously.

\section{Problem formulation} \label{sec:model}
 
The phaseless blind STFT  of a signal $x\in\CN$ with respect to a  window $w\in\CW$, for some $N\geq W$, is given by
\begin{equation} \label{eq:phaseless_blind_stft}
\begin{split}
\left\vert \hat{y}[k,m]\right\vert &= \left\vert\sum_{n=0}^{N-1} x[n]w[mL-n]e^{-2\pi\I kn/N}\right\vert
\\ &=\left\vert x[0] w[mL] + \eta_k x[1] w[mL-1] +
\ldots + \eta_k^{N-1}x[N-1]w[mL-(N-1)]\right\vert,
\end{split}
\end{equation} 
where $\eta_k:=e^{-2\pi\I k/N}$. Here, $0<L<W$ is the step size.
We assume that the signals are zero outside their support so that $w[n]=0$ for all $n\notin[0,\ldots W-1]$ and $x[n]=0$ for  $n\notin[0,\ldots N-1]$. 
The ratio between $W$ and $L$ determines the 
number of short-time sections which is given by $M:=\left\lceil(N+W-1)/L \right\rceil.$ 
The goal is derive the number of frequencies required per window in order
to determine $x$ and $w$ from $\left\vert\hat{y}[k,m]\right\vert$. 
Note that if $W= L$ then  there is no overlap between adjacent sections, and the problem reduces to  the standard phase retrieval problem.

The problem of recovering a signal from its phaseless STFT---when the window $w$ is assumed to be known---was studied thoroughly in recent years~\cite{jaganathan2016stft,eldar2015sparse,bendory2018non,pfander2016robust,iwen2016fast,xu2018accelerated}. However, in ptychography, which is the prime motivation of these papers, the precise structure of the window is usually unknown a priori and thus standard algorithms in the field  optimize over the signal and the window simultaneously~\cite{thibault2009probe,maiden2009improved,guizar2008phase}. 

The main result of this paper, presented in Theorem~\ref{thm:phaselessblindstft}, shows that it suffices to consider less than $10L$ Fourier frequencies per window  to determine 
 a pair of generic signals $(x,w)$. Therefore, in total,  we require $10LM\approx10(N+W)$ phaseless measurements.
 This result is near optimal in the sense that  the number of parameters to be recovered is $2(N+W)$: the real and imaginary parts of the signal and the window. We mention that our result does not hold for the special case when $x=w$ as the problem appears in the FROG setup~\cite{trebino2012frequency}. The latter case was investigated in~\cite{bendory2017signal}, where it was shown that the number of measurements required to determine the (single) signal is three times its bandwidth.  
In Section~\ref{sec:main_results}, we provide a more comprehensive comparison with related results in the literature.  

We also explore the simpler case of blind STFT when the phases are assumed to be known. In this problem, the goal is to determine $x\in\CN$ and $w\in\CW$ from 
\begin{equation} \label{eq:blind_stft}
 \hat{y}[k,m] = \sum_{n=0}^{N-1} x[n]w[mL-n]e^{-2\pi\I kn/N}.
\end{equation} 
For generic $(x,w)$, we show that if $W>L$
then it suffices  to consider only $2L$ Fourier measurements per window to determine $x$ and $w$, up to unavoidable $L$ scaling ambiguities presented in Proposition~\ref{prop:symmetries_phase}. 
Therefore, $2ML\approx 2(N+W)$ measurements are enough. However, in Theorem~\ref{thm:phaselessblindstft} we prove by a more  careful examination  that the number of measurements can be reduced to only $N+W-L$ Fourier measurements  if we choose the  measurements properly for each window. This result is optimal as it captures precisely the number of parameters to be determined.
 
The model of blind STFT~\eqref{eq:blind_stft} can be understood as a special case of the \emph{blind deconvolution} problem. In particular, let us denote by $\hat{y}_{k}$ all entries of $\hat{y}$ for fixed $k$. Then, we can write~\eqref{eq:blind_stft} as 
\begin{equation*}
\hat{y}_{k} = x_k\ast w,
\end{equation*}
where $x_k[n]:=x[n]e^{-2\pi\I kn/N}$ and $\ast$ denotes convolution. 
Blind deconvolution is a fundamental problem in a variety of applications, including  astronomy, communication, image deblurring, system identification, optics and \rev{structural biology}; see~\cite{jefferies1993restoration,shalvi1990new,abed1997blind,ayers1988iterative,michaeli2014blind,bendory2018toward,bendory2019multi} to name a few.
Following advances in related fields like compressed sensing and phase retrieval, many papers have focused on establishing theoretical foundations for different settings of the blind deconvolution problem; see for instance~\cite{ahmed2014blind,li2017identifiability,kech2017optimal,ling2017blind,zhang2018structured}. In particular,~\cite{kech2017optimal} provides a thorough analysis of bilinear problems in general and, in particular, blind deconvolution.
These works assume a low-dimensional structure of the signals to enforce uniqueness. Inspired by some phase retrieval and channel estimation applications, it was shown in~\cite{walk2017blind,walk2017stable} that a pair of signals  with known autocorrelations can be recovered from their blind deconvolution by a convex program.
Our model deviates from the papers mentioned above as it relies on overlapping  windows to derive uniqueness for generic signals with optimal number of measurements.
	
\section{Main results} \label{sec:main_results}

We begin by studying the symmetries of the blind STFT map~\eqref{eq:blind_stft}, often called \emph{trivial ambiguities}. Next, 
we derive the number of  measurements required to determine a pair of \emph{generic}  signals  $(x,w)$.
By generic signals, we mean the following:
\begin{defn} \label{def:generic}
	When saying that a generic signal is uniquely determined by a collection of \rev{polynomial} measurements we mean that the set of signals which cannot be determined by these measurements lies in the vanishing locus of a nonzero polynomial\footnote{Given a  nonzero polynomial $f$ in $T$ variables (real or complex), the vanishing locus of $f$ is the set 
	$V(f) = \{(x_1, \ldots, x_T) \in \R^T (\text{resp. } \C^T) \vert f(x_1,\ldots x_T)=0\}$. By a general result in algebraic geometry $\dim V(f) < T$. As a consequence, the complement of $V(f)$ in $\R^T$ (resp. $\C^T$) is dense.}. In particular, this means that we can recover almost all signals with the given measurements.
\end{defn}

\rev{We remark that while almost all signals are generic, Definition~\ref{def:generic} does not cover  several classes of signals which might be important in some applications. A notable example is the class of sparse signals (signals with only few nonzero entries). 
For instance, the following pair of signals $(x_1,w_1)$ and $(x_2,w_2)$ share the same blind phaseless STFT~\eqref{eq:phaseless_blind_stft} (this example is taken from~\cite{choudhary2014fundamental}): 
\begin{align*}
x_1 =& [1, 0, 1, 0, 0, 0, 0, 0, 1, 0, 1],  &w_1 = [1,0,0,0,1,0,0], \\
x_2 =& [1, 0, 0, 0, 0, 0, 0, 0, 1, 0, 0],  &w_2 = [1,0,1,0,1,0,1]. 
\end{align*}
Further ambiguities for sparse signals when the window is known are discussed in~\cite{eldar2015sparse}. 
}

Since the window $w\in\CW$ is unknown,  the group of ambiguities of the map~\eqref{eq:blind_stft} is large, although the phases of the blind STFT are known. 
This group of ambiguities can be thought of as   \emph{scaling ambiguities} between the signal and the window  as formulated in the following proposition. 
\begin{proposition} \label{prop:symmetries_phase} 
	Let $\lambda:=(\lambda_0, \ldots, \lambda_{L-1})$ be complex, nonzero, numbers. 
	Then, the following action preserves the blind STFT measurements~\eqref{eq:blind_stft}:
	\begin{eqnarray*}
		\lambda\circ (x,w)  := 
		\left((\lambda_0 x[0], \ldots , \lambda_{L-1} x[L-1],
		\lambda_0 x[L] ,\ldots  ), ( \lambda^{-1}_0 w[0], \lambda^{-1}_{L-1} w[1], \ldots
		\lambda^{-1}_1 w[L-1], \lambda_0^{-1} w[L], \ldots )\right).
	\end{eqnarray*}
\end{proposition}

\begin{proof}
	 Note that  
	$\lambda$ acts on $x$ by multiplication by $\lambda_{n\mod L}$ and on $w$ by multiplication by $\lambda^{-1}_{(L-n) \mod L}$. 	
	Thus, every term in the sum for $\hat{y}[k,m]$ is preserved under the action of
	$\lambda$ since
	$x[n] w[mL -n]$ is transformed into $(\lambda_{n \mod L}) x[n] (\lambda^{-1}_{n \mod L}) w[mL -n]$.	
\end{proof}
\noindent Note that Proposition~\ref{prop:symmetries_phase} can be understood as an action of a group  $(\C^\times)^L$ on $\C^N \times \C^W$.

Next, we are ready to present the uniqueness result for the blind STFT case.
We show that for each section it suffices to acquire $2L$ Fourier measurements to determine the signals. Thus, $ML\approx 2(N+W)$ measurements are enough to determine the signals up to the $L$ scaling ambiguities of Proposition~\ref{prop:symmetries_phase}.
This \rev{bound} can be improved  by choosing the number of Fourier measurements carefully per section. Particularly, we show that  only $N+W-L$ measurements suffice. This result is optimal since the number of parameters to be determined is $N+W$ (the length of the signals), while $L$ entries can be fixed arbitrarily by Proposition~\ref{prop:symmetries_phase}. 
Interestingly, the result holds for any $L<W$. Therefore, considering small overlaps between adjacent windows (i.e., large $L$) increases only the size of the ambiguity group.
The proof is given in Section~\ref{sec:proof_uniqueness_phaseless} and is based on a recursive argument.

\begin{thm} \label{thm:uniqueness_withphases}
For a generic signal $x\in\CN$ and a  generic window $w\in\CW$ with $W> L$, the pair $(x,w)$ is uniquely determined, modulo the trivial ambiguities of Proposition~\ref{prop:symmetries_phase}, from  $N+W-L$ measurements from~\eqref{eq:blind_stft}. \end{thm}

While the model~\eqref{eq:blind_stft} has not been studied in the literature, it is instructive to compare  Theorem~\ref{thm:uniqueness_withphases} with~\cite{kech2017optimal}. In this paper, Kech and Krahmer considered the uniqueness of bilinear maps and, as an application, blind (circular) convolution maps $\CN\times\CN\to \CN$. 
In particular, they assume  the  two signals lie in generic known subspaces of dimension $d_1$ and $d_2$. In this case, \emph{all signals} are uniquely determined from their convolution provided that $N\geq 2(d_1+d_2)-4$, modulo a one-dimensional scaling ambiguity between the two signals. 
A similar result holds true when the signals are sparse (see also~\cite{li2016identifiability,choudhary2014sparse}). 
Therefore, the number of measurements is approximately twice the number of parameters to be estimated. 
When considering generic sparse signals, it suffices to demand  $N\geq s_1+s_2$, where $s_1$ and $s_2$ denote the cardinality of the signals.
Comparably, Theorem~\ref{thm:uniqueness_withphases} states that  the number of measurements required to determine a pair of generic signals is  exactly the number of parameters to be recovered $N+W-L$. We did not derive a result that holds for all signals. In addition, 
\rev{the dimension of the ambiguity group presented in Proposition~\ref{prop:symmetries_phase} grows with $L$, whereas the dimension of the ambiguity group in~\cite{kech2017optimal} is always one.}

We now turn our attention to the problem of determining a pair of signals from their phaseless blind STFT measurements~\eqref{eq:phaseless_blind_stft}, which is the focal point of this paper. As expected, in this case the group
of trivial ambiguities increases, but only by two real dimensions. These two additional symmetries correspond to multiplication by global phase and continuous modulation. We note that other symmetries that frequently appear in phase retrieval setups, such as conjugate reflection or discrete shifts, do not occur in our setting because of the aperiodicity of the setup; compare with~\cite[Proposition 1]{bendory2017uniqueness}.  In what follows, $S^1$ denotes the unit circle:
 $$S^1:=\{\eta\in\mathbb{C}: \vert\eta\vert =1\}.$$  

\begin{proposition} \label{prop:symmetries_phaseless}
		Let $\lambda:=(\lambda_0, \ldots, \lambda_{L-1})$ be complex, nonzero, numbers and let $\eta_1,\eta_2\in S^1$.  
	    Then, the following 
		action preserves the phaseless blind STFT measurements~\eqref{eq:phaseless_blind_stft}:
		\begin{align*}
		(\eta_1, \neta_2,\lambda) \circ (x,w)  = & 
		(\eta_1 \lambda_0 x[0], \eta_1 \eta_2 \lambda_1 x[1] , \ldots , \eta_1
		\neta_2^{L-1}\lambda_{L-1} x[L-1],
		\eta_1  \neta_2^L\lambda_0 x[L] , \ldots  ,\\ &
		\lambda^{-1}_0 w[0],  \neta_2\lambda^{-1}_{L-1} w[1], \ldots
		, \eta_2^{L-1}\lambda^{-1}_1 w[L-1], \eta_2^L\lambda_0^{-1} w[L], \ldots ).
		\end{align*}	
		If $N=W$, then we have an additional ambiguity because
		$(x,w)$ and $(\overline{w}, \overline{x})$ have the same phaseless blind  STFT measurements.
\end{proposition}
\begin{proof}
		It is straightforward to verify that under this action, the blind STFT measurement $\hat{y}[k,m]$ is translated to $\eta_1 \eta_2^{mL}\hat{y}[k,m]$ which has the same absolute value.		
\end{proof}

As for the blind STFT ambiguities, Proposition~\ref{prop:symmetries_phaseless} can be understood as an action of the group  $S^1 \times S^1  \times (\C^\times)^L$ on $\C^N \times \C^W$. 
The following proposition uncovers an interesting property of this action. Specifically, it 
shows that the unique element of a quotient of this group that maps {a generic vector}  to itself is the identity. The proof is provided in Section~\ref{sec:proof_generically_freely}.
\begin{proposition} \label{prop:generically_freely}	  
	Define an action of
   the group $ S^1 \times S^1  \times (\C^\times)^L$  on $\C^N \times \C^W$ as in Proposition~\ref{prop:symmetries_phaseless}.	Then,  a quotient of this group by the finite group $\mu_L$
	of $L$-th roots of unity
	acts generically freely\footnote{See Appendix~\ref{sec:groups} and in particular Example~\ref{ex:app}.}.
\end{proposition}

We are now ready to present the main result of this paper. As in \rev{Theorem~\ref{thm:uniqueness_withphases}}, the result holds true for any $W>L$. In particular, for any window it is sufficient to acquire less than $10L$ Fourier intensity measurements. Therefore, less than $10LM\approx 10(N+W)$ are required in total.
Increasing $L$  increases the size of the ambiguity group but has only a negligible effect on the number of required measurements.
The proof is provided in Section~\ref{sec:proof_main_thm}. 

\begin{thm} \label{thm:phaselessblindstft}
	For a generic signal $x\in\CN$ and a generic window $w\in\CW$, the pair $(x,w)$ is uniquely determined, modulo the trivial ambiguities of Proposition~\ref{prop:symmetries_phaseless},
	from fewer than $10(N+W)$ phaseless blind STFT measurements~\eqref{eq:phaseless_blind_stft}. Precisely, $(x,w)$ is uniquely determined by the measurements
	\begin{equation*}
\{|\hat{y}[0,m]|, \ldots , |\hat{y}[Q(m),m]|\}_{m=0, \ldots , M-1},	\end{equation*}
where $Q(m) := \min(2mL+1, 10L-3)$ and $M:=\left\lceil(N+W-1)/L \right\rceil$.  
\end{thm}

In~\cite{bendory2017uniqueness}, it was shown that a pair of a signal and a window (which, in contrast to our model, are allowed to be equal to each other) can be uniquely determined when $L=1$ and all $N$ Fourier intensity measurements are recorded for each window, resulting in order of $N^2$ phaseless blind STFT measurements.
\rev{Namely}, the number of measurements is quadratic in the number of parameters to be determined.
Theorem~\ref{thm:phaselessblindstft} improves this result significantly as,  for any $L$, the number of required measurements is linear in $N$.
In~\cite{bendory2017signal}, we have shown that in the special $x=w$ case, appearing in the FROG technique~\cite{trebino2012frequency}, one needs to acquire only $3B$ measurements, where $B$ is the bandwidth of the signal.
The  two-dimensional blind ptychography problem was analyzed in~\cite{fannjiang2018blind}.
It is important to note, however, that in contrast to the one-dimensional case, the phase retrieval problem \rev{admits a unique solution 
 for generic signals in two dimensions}\footnote{\rev{Interestingly, it was recently shown that this solution might be extremely sensitive to errors~\cite{barnett2018geometry}.}
}, up to ambiguities~\cite{hayes1982reconstruction,bendory2017fourier}. 

Finally, we would like to refer to a recent paper~\cite{ahmed2018blind}, considering the recovery of a pair of signals from the Fourier magnitudes of their (circular) blind deconvolution. The underlying assumption is that the two signals lie in low-dimensional random subspaces of dimensions $k$ and $m$.
Thus, it studies a complementary problem to our model which is based on overlapping windows~\eqref{eq:phaseless_blind_stft}.
The main result of this paper states that the two signals can be 	recovered by a convex program provided that $N/\log^2 N\gg (k+m)$.

\section{Phase retrieval from limited measurements for partially known signals} \label{sec:partial_known_signal}

In this section, we study a general question in phase retrieval about recovering a signal from its Fourier magnitudes, where a subset of the signal's entries  is already known. 
This situation occurs in the phaseless blind STFT problem since, if $W> L$, the sections in~\eqref{eq:phaseless_blind_stft} overlap. Thus, the $m$th section is recovered from its Fourier magnitudes and the knowledge of some of its entries. These known entries 
are  determined by the $(m-1)$th section. We use this procedure successively in the proof of Theorem~\ref{thm:phaselessblindstft} in Section~\ref{sec:proof_main_thm}. 
However, since the main result of this section, Proposition~\ref{prop:limited}, is quite general, we devote this separate section to present it and discuss its ramifications. 

Consider the following setup. Let $x\in \CN$ be a signal and, as usual, we assume $x[i] = 0$ for $i \notin \{0, \ldots ,
N-1\}$.  
In this section, we {examine} the 
continuous Fourier measurements $S^1 \to \R_{ \geq 0}$  given by 
\begin{equation}
A(\eta) =\left\vert \sum_{n=0}^{N-1} x[n] \eta^n\right\vert^2,
\end{equation}
where $\eta$ in on the unit circle $S^1$. 
Let $S \subset [0,N-1]$
be a proper subset  (i.e., $\vert S\vert<N$) and assume that $x[n]$ is known for the complementary set $n \in S^c$.

The question of determining a signal from its Fourier magnitudes, when partial information on the signal's entries are known, has been studied in~\cite{beinert2017enforcing,beinert2015ambiguities,xu1987almost}, where the former generalizes the results of the two latter papers. In particular, it was shown in~\cite{beinert2017enforcing} that almost any signal is uniquely determined from its Fourier magnitudes and one  entry. Since 
we  use this result repeatedly, we cite it as follows:
\begin{lem} \label{lem:beinert} \cite{beinert2017enforcing}
	Let $\ell$ be an arbitrary integer between 0 and N-1. Then, almost every 
    $x\in\CN$ can be uniquely recovered from $A(\eta)$ and $x[\ell]$. If $\ell=(N-1)/2$, then the reconstruction is up to conjugate reflection.
\end{lem}
\noindent In this section, we generalize this result and derive a bound on the number of Fourier intensity measurements required if several entries of the signal are known.

For a generic $x$, the trigonometric polynomial $A(\eta)$ is of degree $N-1$ and thus can be recovered from its value at $2N-1$ distinct frequencies.  If in addition one entry of the signal is known, then Lemma~\ref{lem:beinert} implies that a generic signal is determined uniquely.
The following proposition, which is the main result of this section, provides a bound  on the number of intensity measurements required for uniqueness as a function of the number of known entries. In particular, if the size of the set of unknown entries $\vert S\vert$ is small relative to $N$, then $x$ can be recovered from only a few measurements. In what follows, we denote a difference set by $$S-S = \{ m-n\,\vert\,  m,n \in S\}.$$ \noindent The following result is proved in Section~\ref{sec:proof_prop_pr}.
 
\begin{proposition} \label{prop:limited}
	If $k \geq  2|S-S|-1 + 2\SV$ and $ \vert S\vert<N$, then
	$A(\eta)$ is uniquely determined by $A(\eta_1), \ldots A(\eta_k)$
	for distinct frequencies $\eta_1, \ldots , \eta_k$. 
\end{proposition}
\noindent Combining Lemma~\ref{lem:beinert} and Proposition~\ref{prop:limited} we conclude:
\begin{corollary} \label{cor:limited}
	If $k \geq  2|S-S|-1 + 2\SV$  and $\vert S\vert<N$ \rev{(that is, at least one signal entry is known)}, then almost every $x\in\CN$ is determined uniquely  by $A(\eta_1), \ldots A(\eta_k)$
	for distinct frequencies $\eta_1, \ldots , \eta_k$.
\end{corollary}

The main application for Proposition~\ref{prop:limited} and Corollary~\ref{cor:limited} in this paper is the proof of Theorem~\ref{thm:phaselessblindstft} in Section~\ref{sec:proof_main_thm}.
In particular, the following situation occurs in phaseless blind STFT. The set of unknown entries is $S = [0, L-1] \bigcup [L(m-1), mL -1]$ and thus $|S| = 2L$.
In this case, the difference set is given by $$S-S = [0,L-1] \cup [L(m-2) +1, mL -1],$$
and hence $|S-S| = 3L-1$. 
Proposition~\ref{prop:limited} affirms that in this case the signal can be determined by the values of at most
$k = 2(3L-1)-1+2(2L)=10L-3$ frequencies.

\section{Proofs} \label{sec:proofs}

\subsection{Proof of Theorem~\ref{thm:uniqueness_withphases}}  \label{sec:proof_uniqueness_phaseless}

Throughout the proof we assume that $w[0], \ldots , w[L-1]$ are all nonzero as $w$ is generic. Because of the scaling ambiguities of Proposition~\ref{prop:symmetries_phase}, we can rescale and assume
that $w[0] = w[1] \ldots w[L-1] = 1$ and thereby eliminate the ambiguities.

Let $x',w'$ be a solution to the system of bilinear equations~\eqref{eq:blind_stft}.
 We will use  recursion to show that for generic $(x,w)$, there is a unique solution under the constraint $w'[0] = w'[1] = \ldots = w'[L-1] = 1$, 
and therefore, $x=x'$ and $w=w'$. 

\textbf{Step 0:}
From~\eqref{eq:blind_stft} we have $\hat{y}[0,0]= x'[0] w'[0]$ (for any $k$).  Since $w'[0] =1$, we can determine $x'[0]$ uniquely. 

\textbf{Step 1:} For $m=1$ and fixed $k$, we get
\begin{equation} \label{eq.step1}
\hat{y}[k,1] = x'[0]w'[L] + \eta_k x[1]'w'[L-1] + \ldots + \eta_k^{L}x'[L] w'[0].
\end{equation}
Since $x'[0]$ is determined in Step 0 and $w'[0], \ldots , w'[L-1]$
are normalized to $1$, we get a linear system of equations with the $L+1$ unknowns  $x'[1], \ldots , x'[L], w'[L]$. 
Hence, if we take $M_1 =L+1$ Fourier measurements, we can uniquely determine
the unknowns for generic $x[0]=x'[0]$.

\textbf{Step 2:}
For $m=2$ and fixed $k$, we get
\begin{equation} \label{eq.step2}
\hat{y}_[k,2] = x[0]' w'[2L] + \eta_k x[1]'w'[2L-1]+ \ldots \eta_k^{L} x'[L] w'[L]
+ \eta_k^{L+1} x'[L+1] w'[L-1]+ \ldots + \eta_k^{2L} x'[2L]w'[0].
\end{equation}
By the previous steps, $w'[0], \ldots w'[L]$ and $x'[0], \ldots , x'[L]$
are determined, so this is a linear equation in
the $2L$ unknowns
$w'[2L], \ldots w'[L+1], x'[L+1], \ldots , x'[2L]$.
Hence we can solve the system with $M_2=2L$ Fourier measurements. In total, up to this stage we require $3L+2$ Fourier measurements.

\textbf{Step m:} In the $m$th step, we  observe 
\begin{eqnarray*}
	\hat{y}[k,m] & = & x'[0] w'[mL] + \eta_k x'[1] w'[mL-1] + \ldots \eta_k^{L-1} 
	x'[L-1]w'[mL -L +1] +\\
	& & 
	\eta_k^L x'[L] w'[mL -L] + \ldots \eta_k^{2L-1}x'[2L-1] w'_[mL-2L+1] + \\
	& & \ldots\\
	& &  \eta_k^{mL-L+1} x'[mL-L +1] w'[L-1] + \ldots \eta_k^{mL} x'[mL]w'[0].
\end{eqnarray*}
By induction, there are at most $2L$ unknowns in this equation,
namely, $x'[L(m-1) + 1], \ldots , x'[Lm]$
and
$w'[Lm], \ldots , w'[Lm - L + 1]$. Hence $M_m= 2L$ Fourier measurements suffice. 

The crude estimate above leads to order of $2ML\approx 2(N+W)$ Fourier measurements. This estimate can be significantly refined by taking into account that we know the length of the signals $N$ and $W$. Starting from the $m=3$, we need one Fourier measurement for each unknown.
Therefore, to determine the last $N-(2L+1)$ entries of the signal we need $N-(2L+1)$ measurements in total, and similarly for the last $W-(2L+1)$ entries of the window (if $W\geq 2L+1$). 
Together with the $3L+2$ measurements required in the first three steps, we conclude that merely $N+W-L$ measurements suffice. The same holds true if $2L \geq W>L$.

\subsection{Proof of Proposition~\ref{prop:generically_freely}} \label{sec:proof_generically_freely}

Since $x$ and $w$ are generic, we may assume that $x[i]$ and $w[L-i]$ are nonzero for  $0\leq i \leq L$.  Then, if $(\eta_1,\eta_2,
\lambda_0, \lambda_{L-1})$ acts trivially on $(x,w)$, then we
see that for all $0 \leq i \leq L$,
$$\eta_1 \neta_2^i
	\lambda_{ (i \bmod L)} = 1 \text{ (coefficient of } x[i]\text{)} $$ and
	$$\neta_2^{mL-i} \lambda_{(i \bmod L)}^{-1} = 1 \text{ (coefficient of }
	w[mL-i]\text{)}.$$
	Looking at the coefficients of $x[0]$ and $x[L]$ we see that
	both $\eta_1 \lambda_0 =1$ and $\eta_1 \eta_2^L \lambda_0=1$ which
	forces $\eta_2^L = 1$.
	Looking at the coefficient of $w[L]$, we see that $\eta_2^{mL} \lambda_0^{-1} = 1$ so $\lambda_0 = \eta_2^L = 1$ and also forces $\eta_1 = 1$.
	Looking at the coefficient of $x[i]$ or $w[i]$ for $i > 0$ we also
	see that $\lambda_i = \eta_2^{mL-i}$. Therefore, the  elements
	of $S^1 \times S^1 \times (\C^\times)^L$ that fix the generic element
	of $\C^N \times \C^W$ are of the form
	$((1, \omega), (1, \omega, \ldots , \omega^{L-1}))$ where
	$\omega = e^{2 \pi \iota k/L}$. These elements form a subgroup
	isomorphic to the finite group $\mu_L$ of $L$-th roots of unity.
	Hence, the group $\left(S^1
	\times S^1 \times (\C^\times)^L\right)/\mu_L$ acts generically freely. Note
	that since
	$S^1/\mu_L$ is isomorphic to $S^1$ and $\C^\times/\mu_L$
	is  isomorphic to $\C^\times$, the quotient group is also
	isomorphic to $S^1 \times S^1 \times (\C^\times)^L$.

\subsection{Proof of Theorem~\ref{thm:phaselessblindstft}} \label{sec:proof_main_thm}

Given a pair of signals $(x',w')$ whose STFT intensity measurements are $\{\vert \hat{y}[k,m]\vert \}$ as in~\eqref{eq:phaseless_blind_stft},  let us define the vector $y_m\in\mathbb{C}^{mL+1}$:
\begin{equation*}
y_m[n] := x'[n] w'[mL-n], \quad n=0,\ldots,mL. 
\end{equation*}
Note that the length of the vectors changes for different values of $m$.
From Theorem~\ref{thm:uniqueness_withphases}, we know that modulo
the $(\C^\times)^L$ group of scaling ambiguities (see Proposition~\ref{prop:symmetries_phase}), a generic pair $(x',w')$
can be recovered from $y_m$ for $m=0, \ldots M-1$. Therefore, we will prove that the vectors $y_m$ are unique, modulo ambiguities.

The proof takes inspiration from the recursive technique of~\cite{bendory2017signal}.
In particular, for fixed $m$, let $A_m(\eta) = \vert \sum_{n=0}^{mL} y_m[n] \eta^n\vert^2$ be the Fourier intensity function of $y_m$, where $\eta\in S^1$.  Then, the  phaseless STFT measurements are the samples 
$\vert\hat{y}[k,m]\vert^2 = A_m(e^{2\pi \I k/N})$.
Since $A_m(\eta)$ is a trigonometric polynomial of degree $mL$, a priori, it can always be recovered from the $2mL+1$ STFT intensity measurements
$\vert\hat{y}[0,m]\vert, \ldots , \vert\hat{y}[0,2mL]\vert$. 
However, we will show that if $y_0, \ldots , y_{m-1}$ are previously known (modulo ambiguities),  then
$A_m(\neta)$ can be recovered from at most $10L-3$ STFT intensity measurements and that the knowledge of $y_{m-1}$ and $A_m(\neta)$ uniquely determines $y_m$. 

The proof follows the  recursion below:

\textbf{Step 0:}
As in Section~\ref{sec:proof_uniqueness_phaseless}, we may use the action of
the ambiguity group to assume that $w'[0] =\ldots = w'[L-1] =1$.
Since we know $|\hat{y}[0,0]| = |x'[0] w'[0]|$ and  we assume that $w'[0] = 1$,
we know $x'[0]$ up to phase. Using the global phase ambiguity, we can then fix the phase of $x'[0]$ arbitrarily.

\textbf{Step 1:} The trigonometric polynomial $A_1(\neta)$ has degree $L$,
so it is uniquely determined by the  $2L+1$ STFT intensity measurements
$\vert\hat{y}[0,1]|, \ldots |\hat{y}[2L,1]\vert$. However,
the vector $y_1 = (x'[0]w'[L], \ldots , x'[L] w'[0])$
is not uniquely determined by the Fourier intensity function $A_1(\eta)$. Indeed, there are up to $2^{L-1}$ vectors, modulo global phase and reflection and conjugation,
with Fourier intensity function $A_1(\eta)$; see for instance~\cite[Section 3.1]{bendory2017fourier}. 

\textbf{Step 2:}
Using the second $S^1$ ambiguity, we fix the phase of $w'[L]$ arbitrarily. For example, we may fix it to be real. 
In addition, note that for each of the possible solutions from Step 1,
$\vert w'[L] x'[L]\vert$ has to be the same. 
The reason is that the last entry of the autocorrelation of  $y_1$ (which is equivalent to $A_1(\eta)$) is given by 
$w[0]\overline{x[0]} w[L] \overline{x[L]}$  and $x[0]$ and $w[0]$ are already known from Step 0. 

Now, $A_2(\eta)$ is a trigonometric polynomial of degree $2L$,
hence it is uniquely determined by $4L + 1$ STFT intensity measurements
$\vert\hat{y}[0,2]|, \ldots , |\hat{y}[4L,2]\vert$.
From Lemma~\ref{lem:beinert}, we know that the Fourier magnitudes of a signal and  the knowledge of one of its entries suffice to determine almost all signals. 
Since $\vert x'[L]w'[L]\vert$ is determined, for each $(x'[0], \ldots , x'[L], w'[0], \ldots , w'[L])$ obtained in Step 1,
there are two possible vectors 
\begin{equation*}
(x'[0], \ldots , x'[2L], w'[0], \ldots , w'[2L])  \mbox{  and  } (\overline{w'}[0], \ldots , \overline{w'}[2L], \overline{x'}[0], \ldots
, \overline{x'}[2L]),
\end{equation*}
with Fourier intensity function determined by the phaseless STFT measurements
$|\hat{y}[0,2]|, \ldots , |\hat{y}[4L, 2]|$.

\textbf{Step 3:}
After step 2, we have $2^{L}$ possible values for the unknowns
$$(x'[0], \ldots , x'[2L], w'[0],\ldots, w'[2L]),$$ that are consistent with the constraints 
\begin{equation} \label{eq:step3_1}
\left\vert x'[0]w'[L] + \neta x'[1] w'[L-1] + \ldots \neta^{L}x'[L]w'[0]\right\vert^2 = A_1(\eta),
\end{equation}
and
\begin{equation} \label{eq:step3_2}
\left\vert x'[0] w'[2L]+ \eta x'[1]w'[2L-1] + \ldots \eta^{2L}x'[2L]w'[0]\right\vert^2 = A_2(\eta).
\end{equation}

In this stage, $A_3(\eta)$ is a trigonometric polynomial of degree $3L$ so it can be determined by  the
$6L+1$ STFT intensity measurements $|\hat{y}[0,3]|, \ldots, |\hat{y}[6L,3]|$. Moreover, for each choice
of $(x'[0], x'[1], \ldots, x'[2L], w'[0], \ldots , w'[2L])$, the values of
$w'[2L+1], \ldots , w'[3L], x'[2L+1], \ldots , x[3L]$ are uniquely determined
from $A_3(\eta)$; see Lemma~\ref{lem:beinert}. 

The following result shows that at this step there is only one pair of vectors (modulo ambiguities), out of the $2^L$ possible vectors of Step 2, that is consistent with the constraints. The proof is somewhat technical so we defer it to Appendix~\ref{sec:proof_step3uniqueness}.

\begin{proposition} \label{prop:step3uniqueness}
	For generic choice of vectors $x\in\CN$ and $w\in\CW$, there is a unique choice among the $2^L$
	possible vectors $(x'[0], \ldots , x'[2L], w'[0], \ldots , w'[2L])$ that is consistent with~\eqref{eq:step3_1} and~\eqref{eq:step3_2}  such that the equations
	\begin{gather*} \left\{ \left\vert\sum_{n=0}^{3L} x'[n] w'[3L-n]e^{2\pi \I kn/N} \right\vert= \left\vert \hat{y}[k,m]\right\vert \right\}_{k=0,\ldots , 2L}
	\end{gather*}
	have a solution for the unknown $x'[i]$ and $w'[j]$ for $i,j=0,\ldots 2L$.
\end{proposition} 
From Proposition~\ref{prop:step3uniqueness}, we have  uniquely determined  the vectors
$y_0,y_1, y_2, y_3$, modulo the trivial ambiguities of Proposition~\ref{prop:symmetries_phaseless}, and we can now proceed recursively.  

\textbf{Step m:}
Note that knowledge of $y_0, \ldots , y_{m-1}$ determines, modulo trivial ambiguities, the entries $(x'[0], \ldots , x'[(m-1)L])$ and $(w'[0], \ldots , w'[(m-1)L]$. Then, if $m > 3$ and $y_0, \ldots , y_{m-1}$ are known, then the entries $y_m[n] = x'[n]w'[mL -n]$ are known if $L\leq n \leq  (m-1)L$.
Hence, from Lemma~\ref{lem:beinert}, a pair of generic signals $(x,w)$ is uniquely determined by the Fourier intensity function $A_m(\eta)$.
A priori, we need $2mL +1$ intensity measurements to determine the trigonometric polynomial $A_m(\eta)$. However,
the only unknown entries in the vector $y_m$ are the $2L$ coefficients $S= [1,L] \cup [mL -L, mL -1]$. By proposition~\ref{prop:limited}, we can therefore uniquely determine $A_m(\eta)$ from at most  $10L-3$ measurements.

This concludes the proof.

\subsection{Proof of Proposition~\ref{prop:limited}}
\label{sec:proof_prop_pr}
We begin by expressing $A(\eta)$ as a trigonometric polynomial. Let us write
$x[n] = u[n] + \iota v[n]$, so
$$A(\eta) = \left(\sum_{n=0}^{N-1} u[n]\cos( n\eta) - v[n] \sin(n\eta)\right)^2 + \left( \sum_{n=0}^{N-1} u[n] \sin(n\eta) + v[n] \cos(n\eta)\right)^2.$$
Now, we notice that the right hand side can be expanded in terms of the $4N^2$ quadratic monomials 
$$\{ u[n]u[m], u[n]v[m], v[n]v[m]\}_{\{0 \leq n, m \leq N-1\}}.$$
Writing these expressions explicitly, we get the following coefficients:
\begin{enumerate}
	\item the coefficient of $u[n]u[m] $ is $\cos (n\eta)  \cos (m\eta) + \sin (\eta n) \sin (\eta m)  = \cos((n-m)\eta) = \cos(|n-m|\eta)$;
	\item the coefficient of $v[n]v[m]$ is $\cos(|n-m|\eta)$;
	\item the coefficient of $u[n]v[m]$ is $\sin((n-m)\eta) = \pm(\sin(|n-m|\eta))$.
\end{enumerate}
Since we assume to know the entries for $n \in S^c$,  we denote 
$u[n] = c_n, v[n] = d_n$ and write $x[n] = u[n] + \iota v[n]$ 
only for the unknown set $S$.
For fixed value of $\eta$, we can then write  
the corresponding coefficients explicitly:
\begin{enumerate}
	\item for $u[n]u[m]$ we have $(\SV^2+\SV)/2$  coefficients  given by  $\cos(\vert n-m\vert \eta)$ for $m,n \in S$;
	\item for $v[n]v[m]$ we have $(\SV^2+\SV)/2$ coefficients  given by  $\cos(\vert n-m\vert \eta)$ for $m,n \in S$;
	\item for $u[n]v[m]$ we have $\SV^2-\SV$  coefficients given by  $\sin((n-m)\eta)$      for $m,n \in S$ when excluding $m=n$ (since in this case the coefficients are zeros);
	\item for $u[n]$ we have $\SV$  coefficients given by  $\sum_{m \in S^c} c_m(\cos(n-m) \eta)
	+ d_m \sin((n-m)\eta)$ for $n\in S$;
	\item for $v[n]$ we have $\SV$  coefficients given by $\sum_{m \in S^c} -c_m \sin((n-m)\eta)
	+ d_m (\cos(n-m)\eta)$ for $n\in S$;
	\item the constant term is $\sum_{n,m \in S^c} c_n d_m (\sin(n-m)\eta )+(c_n c_m + d_n d_m) \cos((n-m)\eta)$.
\end{enumerate}
Thus, for any $\eta \in S^1$, the value of $A(\eta)$ provides a linear
equation in the $2\SV^2 + 2\SV$ unknowns
$${\mathcal U} = \{u[n]u[m], u[n]v[m], v[n]v[m]\}_{n,m \in S} \bigcup \{u[n], v[n]\}_{n \in S}.$$
%We will denote this set of equations by $U$.
If we have $k$ intensity measurements $A(\eta_1) , \ldots , A(\eta_k)$,
then we obtain $k$ linear equations in the unknowns ${\mathcal U}$.
However, we note that the coefficients
of the quadratic terms $$\{u[n]v[m], v[n]v[m], u[n]v[m]\}_{\{n,m \in S\}}$$
are in the set $\{\cos(n-m), \pm \sin(n-m)\}_{n-m \in S-S}$.
Therefore, the rank of the linear system is  at most $ 2\vert S-S\vert -1 + 2\SV$, where we subtract one because  $\sin(n-m)=0$ for $m=n$.
It follows that if $k \geq  2\vert S-S\vert -1 + 2\SV$
and $A(\eta_1), \ldots , A(\eta_k)$ are known for distinct
$\eta_1, \ldots , \eta_k$, then $A(\eta)$ can be
computed for any value of $\eta$. Lemma~\ref{lem:beinert} implies that in this case, a generic $x$ can be determined uniquely. 

To explicitly compute $A(\eta)$ from $A(\eta_1), \ldots A(\eta_k),$ we consider
the system of linear equations 
in the variables ${\mathcal U}$:
\begin{eqnarray} \label{eq:linear_system}
	f(\eta_1)  & = & A(\eta_1) - D(\eta_1)\nonumber\\
	f(\eta_2)  &  = & A(\eta_2) - D(\eta_2)\\
	\vdots \nonumber \\
	f(\eta_k) & = & A(\eta_k) - D(\eta_k), \nonumber
\end{eqnarray}
where for any value $\eta$
\begin{align*}
f(\eta):=&\sum_{n,m \in S} (\cos\eta(n-m)) (u[n]v[m] +v[n]v[m]) + \sin(\eta(n-m)) u[n]v[m]  \\+ &
\sum_{n \in S,m \in S^c} (\cos(\eta(n-m))u[m] + 
\sin(\eta(n-m)) v[m]) u[n]  \\+&
\sum_{n \in S,m \in S^c} (-\sin(\eta(n-m))u[m]  + \cos(\eta(n-m)) v[m]) v[n])
\end{align*}
is a linear combination of the variables in the set ${\mathcal U}$
and 
$$D(\eta)  := \sum_{n,m \in S^c} c_n d_m (\sin(n-m)\eta )+(c_n c_m + d_n d_m) \cos((n-m)\eta)$$
is known.
\rev{Accordingly, the right-hand side of~\eqref{eq:linear_system} contains $k$ known scalars, and the left-hand side is composed of a linear combination  of the unknown set ${\mathcal U}$; the particular linear combination depends on the signal's known entries (the set $S^c$) and the sampling points $\eta_1,\ldots,\eta_k$.}
Since the rank of the coefficient matrix of this system is at most
$  2\vert S-S\vert -1 + 2\SV$ we know that it is overdetermined.
Hence, \rev{for fixed $\eta$}, there exist (not necessarily unique) scalars $\lambda_1, \ldots ,
\lambda_k$ such that
$$f(\eta) = \lambda_1 f(\eta_1) + \ldots \lambda_k f(\eta_k).$$
We also know that the  system is consistent since $A(\eta_1), \ldots , A(\eta_k), A(\eta)$
are Fourier intensity measurements from an actual vector. %$ (u[0] + \I v[0], \ldots , u[N-1] + \I v[N-1])$. 
It follows that
$$A(\eta) - D(\eta) = \lambda_1 \left(A(\eta_1) - D(\eta_1)\right) +
\ldots  + \lambda_k
\left(A(\eta_k) - D(\eta_k)\right),$$
which allows us to determine $A(\eta)$ since $D(\eta)$ is known.

\rev{An important question from a practical point of view is whether the linear system~\eqref{eq:linear_system} is stable to errors. We leave this analysis for future research}.

\section*{Acknowledgment}
Tamir Bendory is grateful to Ti-Yen Lan for discussions on ptychography. 
Dan Edidin acknowledges support from Simons Collaboration Grant 315460. Yonina C. Eldar acknowledges support from the European Union's Horizon 2020 research and innovation program under grant agreement No. 646804-ERCCOG-5BNYQ, and from the Israel Science Foundation under Grant no. 335/14.

\bibliographystyle{ieeetran}
%\bibliographystyle{plain}
%\bibliography{ref}
%\bibliography{BlindSTFT_IT_v2.bbl}

\appendix 

\numberwithin{equation}{subsection}

\subsection{Necessary background on group theory} \label{sec:groups}

For completeness, we briefly present some basic definitions in group theory used throughout the proofs. Most of the definitions are taken from~\cite{chirikjian2016harmonic}, yet they also appear in any standard book on group theory. %In what follows, we merely provide formal and concise definitions and in general do not discuss implications or show examples.

We start with the general definition of a group.
\begin{defn}
A group is a set $G$ together with an operation $\circ$ such that for any $g_1,g_2,g_3\in G$ the following properties hold:
\begin{itemize}
	\item (closure) if $g_1,g_2\in G$, then $g_1\circ g_2\in G$;
	\item (associativity) $g_1\circ(g_2\circ g_3) = (g_1\circ g_2)\circ g_3$;
	\item (existence of an identity) there exists an element $1\in G$ such that $1\circ g= g\circ 1=g$ for any $g\in G$;
	\item (existence of an inverse) for every element $g\in G$, there exists an element $g^{-1}\in G$ such that $g^{-1}\circ g=g\circ g^{-1}=1$.
\end{itemize}  
\end{defn}
\noindent We next present the notion of subgroups and in particular normal subgroups.
\begin{defn}
 A subgroup $H$ is a subset of a group $H\subseteq G$ which is itself a group, namely, it is closed under the group operation of $G$, $1\in H$ and $h^{-1}\in H$ whenever $h\in H$. A subgroup $N$ is called \emph{normal} if  $g^{-1}Ng=N$ for all $g\in G$.  
\end{defn}
\noindent %By definition, normal subgroups are invariants under conjugation. 
For example, any subgroup of an Abelian group (a group satisfying $g_1\circ g_2=g_2\circ g_1$ for all $g_1,g_2\in G$) is normal. 

Our next definition concerns the notions of cosets and equivalence classes:
\begin{defn}
	Given a subgroup $H$ of $G$ and any element $g\in G$, the left coset $gH$ (similarly, the right coset $Hg$) is defined as 
	\begin{equation*}
	gH:=\{g\circ h\vert h\in H\}.
	\end{equation*} 
Any group is divided into disjoint left (right) cosets and if $g_1$ and $g_2$ are in the same coset %, %that is $g_1\circ g_2^{-1}\in H$, 
we denote $g_1\sim g_2$. The relation between $g_1$ and $g_2$ is called an  \emph{equivalence relation}. The set of all left (right) cosets is called \emph{coset space} and is denoted by $G/H$ $(G\backslash H)$.  	
\end{defn}
\noindent Note that  the sets of left and right cosets coincide for normal subgroups. 
We are now ready to define quotient groups:

\begin{defn}
If $N$ is a normal subgroup of $G$, then the coset space $G/N$ together with the operation  $(g_1N)(g_2N)=(g_1\circ g_2)N$ is a group, called the \emph{quotient group}.
\end{defn}

The last notion we would like to mention is of a generically freely action of a group, which is used in Proposition~\ref{prop:symmetries_phaseless}.
A group $G$ acts \emph{generically freely}
on a vector space $V$ if for a generic vector $v \in V$, $gv = g'v$ for some $g,g'\in G$, 
if and only if $g=g'$. Equivalently, $gv =v$ if and only if $g$ is the identity element. In other words, the only element that acts trivially on a generic vector is the identity. Below we provide a simple example to illustrate the definition:

\begin{example} \label{ex:app}
	Let $G=S^1\times S^1$ acting on $\mathbb{C}^4$, where we view $G$ as consisting of diagonal matrices 
	$\begin{pmatrix}
	e^{\I\theta_1} & 0\\0 &e^{\I\theta_2} &
	\end{pmatrix}$ 
	and $\mathbb{C}^4$ as $2\times 2$ matrices 
$\begin{pmatrix}
a & b\\c &d &
\end{pmatrix}$.
The action of $G$ on $\mathbb{C}^4$ is by conjugation; that is, $g\circ x$ is defined to be the matrix 
\begin{equation*}
\begin{pmatrix}
e^{\I\theta_1} & 0\\0 &e^{\I\theta_2} &
\end{pmatrix}
\begin{pmatrix}
a & b\\c &d &
\end{pmatrix}
\begin{pmatrix}
e^{-\I\theta_1} & 0\\0 &e^{-\I\theta_2} &
\end{pmatrix} = 
\begin{pmatrix}
a & be^{\I(\theta_1-\theta_2)}\\ ce^{\I(\theta_2-\theta_1)} &d &
\end{pmatrix}.
\end{equation*}
Now, observe that if $g=	\begin{pmatrix}
	e^{\I\theta} & 0\\0 &e^{\I\theta} &
\end{pmatrix}
$ for any $\theta$, then $g\circ x = x$ for all matrices in $\mathbb{C}^4$.
 Thus, there is a subgroup of $S^1\times S^1$ which acts trivially on $\mathbb{C}^4$, even though it consists of more than
 the identity element. 
 This is an example of an action which is not generically free. 
{On the other hand, the quotient of $S^1\times S^1$ by the subgroup of  matrices where $\theta_1=\theta_2$ acts generically free.}
\end{example}

\subsection{Proof of Proposition~\ref{prop:step3uniqueness}}
\label{sec:proof_step3uniqueness}

Let us define the vectors
\begin{align*}
y_0 &:= (x[0]w[0])\in\mathbb{C}, \\
y_1&:= (x[0]w[L],  \ldots , x[L] w[0])\in\mathbb{C}^{L+1},\\ 
y_2 &:= (x[0]w[2L], \ldots , x[2L]w[0])\in\mathbb{C}^{2L+1},\\
y_3 &:= (x[0] w[3L], \ldots , x[3L]w[0])\in\mathbb{C}^{3L+1}.
\end{align*}
The entries of $y_0,y_1, y_2, y_3$ satisfy the following relations:
\begin{equation} \label{eq:relation_of_y0}
y_0[0]y_2[L] = y_1[0]y_1[L]
\end{equation}
\begin{equation} \label{eq:relation_of_ys}
\{y_1[n] y_3[L+n] = y_2[L+n] y_2[n]\}_{n=0,\ldots L}.
\end{equation}
Equation~\eqref{eq:relation_of_y0}  implies that we may determine $y_0$ from the entries of
$y_1, y_2$ provided that $y_2[L] \neq 0$, which is equivalent to assuming that
$w[L]$ and $x[L]$ are nonzero. Thus, we may ignore $y_0$. 
Therefore, we consider the three vectors $y_1, y_2,y_3$,  where $y_1, y_2$ are unconstrained generic vectors and the entries of $y_3$ depends on $y_1, y_2$ thorugh~\eqref{eq:relation_of_ys}. 

The conclusion of Proposition~\ref{prop:step3uniqueness} follows from
the following more general statement. In what follows, we use
$\tilde{z}$ to denote the reflected and conjugated version of
$z\in\mathbb{C}^{P}$, that is, $\tilde{z}[n]=\overline{z[P-n]}$.

\begin{propositionApp} \label{prop:abstractstep3}
	Let  $y_1\in \C^{L+1}, y_2\in C^{2L+1}$ and $y_3\in \C^{3L+1}$ be generic vectors satisfying~\eqref{eq:relation_of_ys}. Then,  
	any solution $y'_1, y'_2, y'_3$ to the system of equations
	\begin{align*}
	\{y'_1[n]y'_3[L+n]&= y'_2[n]y'_2[L+n]\}_{n=0,\ldots , L}\\
	\left\vert \hat{y'}_1(\eta)\right\vert^2, &= \left\vert\hat{y}_1(\eta)\right\vert^2, \\
	\left\vert\hat{y'}_2(\eta)\right\vert^2 &= \left\vert\hat{y}_2(\eta)\right\vert^2, \\
	\left\vert\hat{y'}_3(\eta)\right\vert^2 &= \left\vert\hat{y}_3(\eta)\right\vert^2,
	\end{align*}
	must obey $y'_1 = e^{i\theta_1} y_1$, $y'_2 = e^{i\theta_2} y_2$, $y'_3 = e^{i \theta_3} y_3$ for some $\theta_1, \theta_2, \theta_3$ or
	$y'_1 = e^{i\theta_1} \tilde{y}_1$, $y'_2 = e^{i \theta_2} \tilde{y}_2$,
	$y'_3 = e^{i \theta_3} \tilde{y}_3$.
\end{propositionApp}
\begin{proof} 
	Let $Z \subset \C^{L+1} \times \C^{2L+1} \times \C^{3L+1}$ be the subvariety defined
	the system of equations~\eqref{eq:relation_of_ys}.
	Note that $Z$ is irreducible because it is the product of $L+1$ irreducible
	hypersurfaces,  where the $n$th hypersurface satisfies
	the equation
	$y_1[n] y_3[n+L]= y_2[n] y_2[n+L]$ in the $\C^4$ with coordinates
	$(y_1[n], y_3[n+L], y_2[n], y_2[n+L]$
	and $n = 0, \ldots , L$.
	
	Notice also that $Z$ is invariant under the action of the group of ambiguities $S^1 \times S^1 \ltimes \mu_2$, where $e^{i\theta_1}, e^{i\theta_2}$ acts
	on $(y_1,y_2,y_3) \in \C^{L+1} \times \C^{2L+1} \times \C^{3L+1}$
	by $(e^{\I\theta_1}y_1, e^{i\theta_2}y_2, e^{2\I\theta_2 - \I\theta_1}y_3)$ and $-1 \in \mu_2$ acts by taking $(y_1, y_2, y_3)$ to $(\tilde{y}_1, \tilde{y}_2, \tilde{y}_3)$.
	The action of $S^1 \times S^1 \ltimes \mu_2$ also preserves the Fourier intensity of each vector. 
	
	Let $X$ be the quotient of $\C^L \times \C^{2L} \times \C^{3L}$ and let $H$ be the quotient of $Z$ under the action of this group. An
	element of $X$ is an equivalence class of triples $(y_1, y_2, y_3)$
	modulo this group of ambiguities and $H$ is the subvariety
	of equivalence classes that satisfy the relations~\eqref{eq:relation_of_ys}.
	
	Consider the
	subvariety $W \subset H \times X$ of pairs of equivalence classes
	$(y_1,y_2,y_3), (y'_1, y'_2,y'_3)$,
	where $y_\ell$ and $y'_\ell$ have the same Fourier intensity function
	for $\ell = 1,2,3$.
	Let $I_W$ be the intersection of $W$ with the subvariety $H \times H$
	of $H \times X$. Both $W$ and $I_W$ clearly contain the diagonal
	(that is, the set of pairs of equivalence classes $\left((y_1,y_2,y_3),
	(y_1,y_2,y_3)\right)$ with $(y_1, y_2, y_3) \in H$).
	
	The projection $W \to H$ (of taking the first triplet) is
	finite because, modulo trivial ambiguities, there are only a
	finite number of vectors with the same Fourier intensity
	function of a given vector (For further discussion of the geometry of related varieties, see~\cite{edidin2018geometry}).
	Let $W_0$ be the closure of the
	complement of the diagonal in $W$ and let $I^0_W$ be the
	closure of the complement of the diagonal in $I_W$.
	
	The projection map $W_0 \to H$ is also finite and
	the assertion of the proposition is equivalent to the statement that
	the image of $I^0_W$ under the projection  is contained in a
	proper (i.e., its complement is dense) subvariety of $H$.
	To see that, note that if
	the equivalence class of $(y_1, y_2, y_3)$ is in the complement of the
	image of $I^0_W$, then if $(y'_1, y'_2,y'_3)$ is a triplet with same
	Fourier intensity function which also satisfies the equations
	\begin{equation} \label{eq:app_sys_equations}
	y'_1[n] y'_3[n+L] = y'_2[n] y'_2[n+L], \quad, n=0,\ldots,L,   
	\end{equation}
	we must have that
	$(y'_1,y'_2, y'_3)$ is obtained from $(y_1, y_2, y_3)$ by an action of
	$S^1 \times S^1 \ltimes \mu_2$.
	
	To this end, it suffices to show that the image of $I_W^0$ is not all of
	$H$. We will show this with an explicit example.
	Let $y_1 = (1, 0, \ldots , 0, 2) \in \C^{L+1}$,
	and let $y_2$ be the vector with $y_2[0] = y_2[L]=y_2[2L] = 1$
	and $y_2[n] = 0$ otherwise.
	Then, the set of vectors
	$y_3 \in \C^{3L}$ 
	that satisfy the relations
	$y_1[n]y_3[n+L] = y_2[n]y_2[n+L]$ is the $(3L-1)$-dimensional linear subspace
	of vectors satisfying 
	\begin{equation} \label{eq:y3}
	y_3[L] = 1 \mbox{ and } y_3[2L] = 1/2.
	\end{equation}
	
	The Fourier transform of $y_2$  is $1 + \eta^L + \eta^{2L}$.  The roots of such a polynomial (called cyclotomic polynomial), are on the unit circle and are given by the $3L$ root of unity.
	Therefore, any vector with the same Fourier intensity must be obtained from $y_2$ by a trivial ambiguity; i.e., multiplication by $e^{i \theta}$ or conjugation and reflection.
	
	The Fourier transform of $y_1$ is the function $1 + 2\eta^L$. Thus, its Fourier intensity function is given by $B(\eta)= 2 \eta^{-L} + 5 + 2 \eta^L$. 
	We also note that  if $y'_1$ is any vector with Fourier intensity function
	$B(\eta)$, then it must have at least 3 nonzero entries unless
	it is of the form $e^{i \theta} y_1$ or $e^{i\theta} \tilde{y}_1$.
	Hence, if $y'_1$ is not equivalent to $y_1$ under the group of trivial
	ambiguities, the solution space to the equations~\eqref{eq:app_sys_equations}
	must have dimension at most $3L-2$, which is inconsistent with~\eqref{eq:y3}.
	This implies that 
	for generic $y_3$ in the linear subspace above there is no non-equivalent $(y'_1, y'_2, y'_3)$ with $y'_3$ having the same Fourier intensity function as $y_3$ and also satisfying the equations~\eqref{eq:app_sys_equations}. 
\end{proof}

\end{document}